\documentclass[journal]{IEEEtran}
\usepackage{amsmath,graphicx,epsfig,latexsym,amsfonts,amssymb}
\usepackage{footnote}
\makesavenoteenv{tabular}

\begin{document}

\title{Full-Duplex Operations in Wireless Powered Communication Networks
        \thanks{This paper has been presented in part at International Conference on Advanced Communication Technologies, Pyeongchang, Korea, July 1-3, 2015.}
        \thanks{The authors are with Electronics and Telecommunications Research Institute (e-mail: jugun@etri.re.kr, yurolee@etri.re.kr, aisma@etri.re.kr).}}

\author{Hyungsik Ju, Yuro Lee, and Tae-Joong Kim}

\maketitle

\IEEEpeerreviewmaketitle
\newtheorem{definition}{\underline{Definition}}[section]
\newtheorem{fact}{Fact}
\newtheorem{assumption}{Assumption}
\newtheorem{theorem}{\underline{Theorem}}[section]
\newtheorem{lemma}{\underline{Lemma}}[section]
\newtheorem{corollary}{Corollary}
\newtheorem{proposition}{\underline{Proposition}}[section]
\newtheorem{example}{\underline{Example}}[section]
\newtheorem{remark}{\underline{Remark}}[section]
\newtheorem{algorithm}{\underline{Algorithm}}[section]
\newcommand{\mv}[1]{\mbox{\boldmath{$ #1 $}}}

\begin{abstract}
In this paper, a wireless powered communication network (WPCN) consisting of a hybrid access point (H-AP) and multiple user equipment (UEs), all of which operate in full-duplex (FD), is described. We first propose a transceiver structure that enables FD operation of each UE to simultaneously receive energy in the downlink (DL) and transmit information in the uplink (UL). We then provide an energy usage model in the proposed UE transceiver that accounts for the energy leakage from the transmit chain to the receive chain. It is shown that the throughput of an FD WPCN using the proposed FD UEs can be maximized by optimally allocating the UL transmission time to the UEs by solving a convex optimization problem. Simulation results reveal that the use of the proposed FD UEs efficiently improves the throughput of a WPCN with practical self-interference cancellation (SIC) capability at the H-AP. With current SIC technologies reducing the power of the residual self-interference to the level of background noise, the proposed FD WPCN using FD UEs achieves 18$\%$ and 25 $\%$ of throughput gain as compared to the conventional FD WPCN using HD UEs and HD WPCN, respectively.
\end{abstract}

\section{Introduction}
As a new way to power mobile devices, there has been a growing interest in harvesting energy from far-field radio-frequency (RF) signal transmissions. In particular, the design of a wireless-powered communication network (WPCN) has been studied as a significant application of RF energy harvesting, where remote user equipment (UE) utilizes the energy harvested from a wireless RF power transfer for wireless communications. A typical WPCN model was proposed in [1], in which a wireless energy transfer (WET) in the downlink (DL) and wireless information transmission (WIT) in the uplink (UL) are both coordinated by a hybrid access-point (H-AP). In a WPCN in which the H-AP and UEs all operate in half-duplex (HD) mode, which is called an HD-WPCNs, a fundamental trade-off exists in allocating resources to the DL for a WET and the UL for a WIT \cite{Ju: HD-WPCN}. This is because allocating more resources to the DL for a WET increases the transmit power of the UEs in the UL thanks to the increase in the amount of harvested energy, while also decreasing the resources allocated to the UL for a WIT.

Meanwhile, full-duplex (FD) based wireless systems, in which the wireless nodes transmit and receive RF signals at the same time and in the same frequency band, have a great deal of attention. FD operation is expected to potentially double the spectral efficiency in wireless communications by cancelling the self-interference (SI) that occurs through the leakage of a transmitted signal that is received by the transmitting node itself. In particular, the feasibility of FD communication has been verified through a proof-of-concept (PoC) design, which shows that the power of the residual SI after self-interference cancellation (SIC) is applied is reduced to a level sufficiently close to that of background noise \cite{KUMU}, \cite{Chea}.

The simultaneous transmission and reception led by an FD operation can also improve the throughput of WPCN. In \cite{Ju: FD-WPCN}, a WPCN model was proposed in which an H-AP with FD operation coordinates a wireless energy transfer to, and a wireless information transmission from a set of UEs operating in HD mode, which is denoted as FD-WPCN with HD UEs (FD-WPCN-HD). It was shown in \cite{Ju: FD-WPCN} that the FD operation of an H-AP in an FD-WPCN-HD is able to increase throughput of a WPCN when the SI is sufficiently cancelled out. In \cite{FD-WPCN2_SumeSun}, the authors studied resource allocation in an FD-WPCN-HD using an energy causality constraint, which is a constraint in which the signal transmission of a UE at a given time can utilize only the energy harvested during a previous time. An FD-WPCN-HD with orthogonal frequency-division multiplexing (OFDM) modulation was furthermore studied in \cite{FD-WPCN3_LeeInKyu}. Finally, in \cite{FD-WPCN4_Rui}, wireless powered relay network was investigated in which an FD relay that utilizes SI as an energy source relays information of an HD source to an HD destination.

For this research, we studied another type of WPCN in which not only the H-AP but also the UEs operate in FD mode, denoted by FD-WPCN with FD UEs (FD-WPCN-FD). Thanks to the FD capability, UEs can receive the signal sent by the H-AP to carry energy while transmitting their respective UL information. We first propose a transceiver structure for UEs that enables simultaneous energy reception in the DL and information transmission in the UL. By extending our previous work \cite{Ju: FD-FD-WPCN}, we characterize the energy usage in the proposed UE transceiver, accounting for an energy leakage from the transmit chain to the receive chain, and show that the UEs can also harvest energy from the leakage of their own UL transmissions (in other words, SI) as well as the received signal sent by the H-AP in the DL. Finally, based on a time-division-multiple-access (TDMA) protocol for WIT in the UL, the optimal time allocation to maximize the weighted sum-throughput in the network, subject to a given total time constraint, is described.

The rest of this paper is organized as follows. Section II presents the system model of FD-WPCN-FD and the proposed UE strucutre. Section III studies throughput of the FD-WPCN-FD and resource allocation to maximize the throughput. Section IV provides simulation results on the throughput of the FD-WPCN-FD. Finally, Section V concludes this paper.

\section{System Model}\label{sec:SystemModel}
Fig. \ref{Fig_SystemModel} shows a FD-WPCN-FD model considered in this paper. This network consists of one H-AP and $K$ UEs (denoted by $U_i$, $i = 1, \cdots, K$), all of which are equipped with a single-antenna FD transceiver to transmit and receive signals over the same frequency band simultaneously.

\begin{figure}[!t]
   \centering
   \includegraphics[width=1.0\columnwidth]{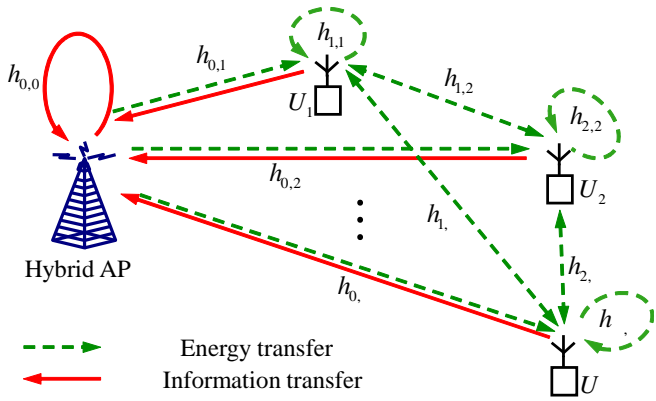}
   \caption{A WPCN model with an FD H-AP and FD UEs.}
   \label{Fig_SystemModel}
\end{figure}

The H-AP transfers energy to the UEs in the DL by broadcasting a signal dedicated to carrying energy. At the same time, the H-AP also receives signals transmitted by the UEs in the UL, which contains information for decoding. The FD operation of the H-AP enables the simultaneous transmission/ reception of energy/information signals in the DL/UL, respectively. For decoding information received in the UL, an SI canceller is deployed at the FD transceiver of the H-AP to eliminate SI (in other words, the leakage of a DL energy signal transmission received by itself), which interferes with the reception of the UL information signals.

On the other hand, UEs harvest energy from the received signals and charge their respective battery. A portion of the charged energy is used to transmit information in the UL. Thanks to the FD operation, UEs are also capable of simultaneously transmitting/receiving information/energy signals in the UL/DL, respectively. This is enabled by the transceiver shown in Fig. \ref{Fig_UE_Transceiver}. Note that the energy signal broadcast by the H-AP in the DL does not carry any information for decoding, but only energy. Therefore, neither an information decoder nor an SI canceller is deployed at the receiving end of a UE's transceiver. Instead, an energy harvester (for example, a rectifier) is installed at the receiving end to harvest energy from the received signals.

It is worth noting that in this transceiver, circulator loss (duplexer loss) at the transceiver yields that a portion of the power amplifier (PA) output energy at the transmitting end leaks into the receiving end as SI, whereas the remaining portion of the PA output energy is used to transmit information to the H-AP. At the $i$-th user $U_i$, as shown in Fig. \ref{Fig_UE_Transceiver}, $E_i$, $E_i^l$, and $E_i^t$ represent the PA output energy at the transmitting end, energy leaking into the receiving end as SI, and energy used to transmit information, respectively, where $E_i = E_i^l + E_i^t$. For convenience, $E_i^l$ and $E_i^t$ are given, respectively, by
\begin{equation}\label{Eq_EnergyUsage}
   E_i^l = {\varphi _i}{E_i} \,\,\,\, {\rm{and}} \,\,E_i^t = (1 - {\varphi _i}){E_i}, \,\,\,\, 0 \le {\varphi _i} \le 1,\footnote{When the transmit and receive antennas of a UE is separated, $\varphi_i$ models the near-field energy absorption by the receiving antenna.}
\end{equation}
where $\varphi_i$ represents the circulator loss at the transceiver of Ui. At the proposed transceiver, there is a trade-off yielded by the value of $\varphi_i$ because as $\varphi_i$ increases, the energy harvested from SI increases whereas the energy used for WIT decreases.

\begin{figure}
   \centering
   \includegraphics[width=1.0\columnwidth]{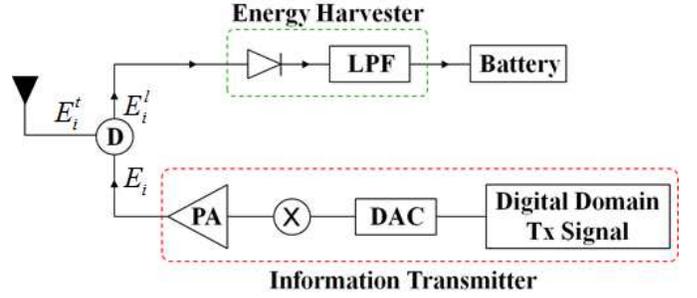}
   \caption{Transceiver structure for FD UEs.}
   \label{Fig_UE_Transceiver}
\end{figure}

In the network shown in Fig. \ref{Fig_SystemModel}, a signal transmission from $U_i$ to $U_j$, $\forall i, j \in \{0, \cdots, K\}$, is assumed to pass through a complex channel $h_{i,j}$ with the channel power gain given by $H_{i,j} = |h_{i,j}|^2 $ ($U_0$  represents the H-AP later in this paper). In particular, when $i = j$, $h_{i,i}$ represents the SI channel at $U_i$ through which the leakage signal of $U_i$ is received by itself. It is assumed that the channels remain constant during a block transmission time denoted by T (in other words, following quasi-static flat-fading), and $h_{i,j} = h_{j,i}$ with the channel reciprocity. The H-AP is assumed to have perfect knowledge of $h_{0,i}$, $\forall i \in \{0, \cdots, K\}$. Furthermore, $\varphi_i$, $\forall i \in \{0, \cdots, K\}$, is also assumed to be known by the H-AP.

Fig. \ref{Fig_FrameStructure} shows a transmission protocol for FD-WPCN-FD. During the whole block duration $T$, the H-AP transmits an energy signal with constant power $P_0$ in the DL, which is received by all UEs. While broadcasting energy by the H-AP, the UEs transmit their own independent information to the H-AP in the UL orthogonally over time through a TDMA. The transmission block therefore consists of $K$ slots, each of which is allocated to $U_i$ for the UL information transmission. The duration of the $i$-th slot is given by $\tau_i T$, $i = 1, \cdots, K$, where
\begin{equation}\label{Eq_SumTime}
   {\sum\limits_{i = 1}^{K} {{\tau _{i}}} \le 1, \,\,\,\, 0 \le \tau_i \le 1. }
\end{equation}

During the $i$-th slot, the H-AP and user $U_i$ transmit a DL signal to carry energy and a UL signal to carry information, respectively. Therefore, the received signals at the H-AP and $U_i$, $i \ne 0$, during the $i$-th slot can be expressed, respectively, as
\begin{equation}\label{Eq_RxSignal_AP}
   y_{i,0} = \sqrt{P_i}h_{0,i}x_i + \sqrt{P_0}h_{0,0}x_0 + n_0,
\end{equation}
\begin{equation}\label{Eq_RxSignal_UE_i}
   y_{i,i} = \sqrt{P_0}h_{0,i}x_0 + \sqrt{P_i}h_{i,i}x_i + n_i,
\end{equation}
with $x_k$, $\forall k = \{0, \cdots, K\}$, denoting the transmitted signal of the $U_k$ with $\mathbb E [|x_k|^2] = 1$, where $\mathbb E [\mu] = 1$ represents the expectation of a random variable $\mu$. Furthermore, $P_k$ represents the transmit power of $U_k$, $\forall k = \{0, \cdots, K\}$. Finally, $n_k$ denotes the receiver noise at $U_k$, $\forall k = 0, \cdots, K$. It is assumed that $n_k \sim \mathcal{CN} (0, \sigma_i^2)$, with $\mathcal{CN} (\nu, \sigma^2)$ representing a circularly symmetric complex Gaussian random variable with mean $\nu$ and variance $\sigma^2$. It is worth noting that the H-AP has to decode $x_i$ in (\ref{Eq_RxSignal_AP}), whereas $U_i$ does not decode information in the DL. At their respective receiving ends, therefore, the H-AP has to cancel the leakage of its own transmission acting as SI (the second term in (\ref{Eq_RxSignal_AP})), whereas it is not necessary for $U_i$ to cancel the leakage of its own transmission (the second term in (\ref{Eq_RxSignal_UE_i})).

\begin{figure}[!t]
   \centering
   \includegraphics[width=1.0\columnwidth]{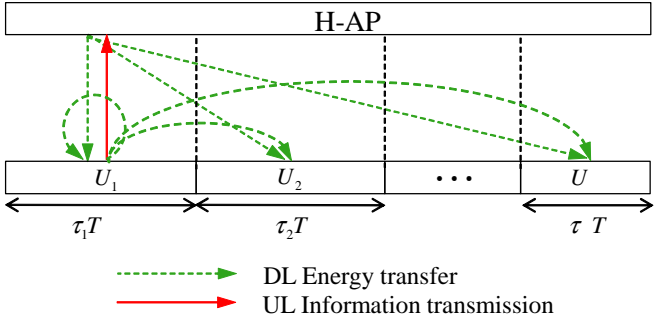}
   \caption{Transmission protocol for FD-WPCN-FD.}
   \label{Fig_FrameStructure}
\end{figure}

In addition, the inactive UEs during the $i$-th slot, in other words, $U_j$ with $j \ne 0$ and $j \ne i$, also receive signals from the transmissions of the H-AP and $U_i$. At an inactive UE $U_j$, $j \ne i$, therefore, the received signal during the $i$-th slot is given by 	
\begin{equation}\label{Eq_RxSignal_UE_j}
   y_{i,j} = \sqrt{P_0}h_{0,j}x_0 + \sqrt{P_i}h_{i,j}x_i + n_j.
\end{equation}

\section{Throughput of FD-WPCN with FD UEs}\label{sec:Throughput}
In this section, the achievable throughput of the FD-WPCN-FD presented in Section \ref{sec:SystemModel} is described. In particular, we aim to maximize the weighted sum-throughput of information transmissions in the UL.

   \subsection{Achievable UL Throughput of Each UE}\label{AchievableRates}
   According to the knowledge of channel available at the H-AP, the throughput of the proposed FD-WPCN-FD can be obtained with two approaches. The first assumes a full knowledge of channels by an aid of a \emph{genie}, whereas the other assumes partial channel knowledge in more practical sense.
   
      \subsubsection{Ideal throughput achievable by a genie-aided approach}
      A UE $U_i$ can harvest energy from the received DL signals transmitted by the H-AP (in other words, the first term in (\ref{Eq_RxSignal_UE_i})) during the whole block duration. In addition, $U_i$ can also harvest energy from the leakage of its own transmission during the $i$-th slot allocated to it for a UL information transmission (in other words, the second term in (\ref{Eq_RxSignal_UE_i})). Furthermore, UL transmissions of other UEs received by $U_i$ can also be used for energy harvesting (EH) of $U_i$. In particular, the energy harvested from other UEs' transmissions can be taken into account to optimize throughput when the H-AP is assumed to have a global knowledge of channels, in other word, $H_{i,j}$, $i, j \in \{0, \cdots, K\}$, $i \ne j$, $i \ne 0$, $j \ne 0$, as well as $h_{0,i}$ and $\varphi_i$, $\forall i \in \{0, \cdots, K\}$. In this case, the harvested energy of $U_i$ during a block duration $T$, denoted by $E_i^h$, can be obtained from (\ref{Eq_RxSignal_UE_i}) and (\ref{Eq_RxSignal_UE_j}) as
      \begin{equation}\label{Eq_Energy_Ui}
         E_i^h
             = \zeta_i \left( {{P_0}H_{0,i} + {\tau _i}{P_i}H_{i,i} + \sum\limits_{\scriptstyle j = 1 \hfill \atop
             \scriptstyle j \ne i \hfill}^K {\tau_j P_j H_{j,i}} } \right)T,
      \end{equation}
      where $0 \le \zeta_i le 1$, $i = 1, \cdots, K$, denotes the energy harvesting efficiency at $U_i$. The energy harvesting efficiency accounts for RF-direct current (DC) conversion efficiency of a rectifier because it represents the ratio of the energy of the received signal (RF) converted into the energy harvested and stored (DC). In (\ref{Eq_Energy_Ui}), assuming that the amount of energy harvested from the receiver noise is small enough to be neglected, such energy is not considered \cite{Ju: HD-WPCN}. In the FD UE transceiver shown in Fig. \ref{Fig_UE_Transceiver}, it is worth noting that, based on the law of energy conservation, the energy harvested from the leakage of its own transmission, in other words, $\tau_i P_i H_{i,i} T$ in (\ref{Eq_Energy_Ui}), is equivalent to the amount of energy leaking from the transmitting end into the receiving end. From (\ref{Eq_EnergyUsage}), we therefore have
      \begin{equation}\label{Eq_LeakageEnergy}
         \tau_i P_i H_{i,i} T = E_i^l = \varphi_i E_i.
      \end{equation}
      
      In each transmission block, a portion of the harvested energy is used to generate the PA output signal at the transmitting end of the transceiver, as shown in Fig. \ref{Fig_UE_Transceiver}. The energy of the PA output signal at $U_i$ is then given by
      \begin{equation}
         E_i = \eta_i E_i^h,
      \end{equation}
      with $0 \le \eta_i \le 1$, $i = 1, \cdots, K$, denoting the portion of the harvested energy used to generate a PA output signal at $U_i$ under a steady state.\footnote{In practical scenarios, energy causality can also be considered as studied in \cite{FD-WPCN2_SumeSun}. To focus on the FD operation at UEs, however, we do not consider the energy causality in this study by assuming that the batteries of $U_i$'s are charged sufficiently at the beginning of each block and $U_i$'s use their respective harvested energy such that the energy charged in the battery at the end of each block reaches its initial state.} Note that $E_i^t$ in (\ref{Eq_EnergyUsage}) can be represented as
      \begin{equation}\label{Eq_E_i_t}
         E_i^t = \tau_i P_i T = \left( 1 - \varphi_i \right) E_i.
      \end{equation}
      From (\ref{Eq_Energy_Ui}) through (\ref{Eq_E_i_t}), we then have 
      \[
         E_i^h
          = \zeta_i \left( { {P_0}H_{0,i} +  \frac{\varphi_i \tau_i P_i }{1 - \varphi_i}
            + \sum\limits_{\scriptstyle j = 1 \hfill \atop \scriptstyle j \ne i \hfill}^K {\tau_j P_j H_{j,i}} } \right) T
      \]
      \begin{equation}\label{Eq_Energy_Ui_new2}
         = \frac{\tau_i P_i T}{\eta_i \left( 1 - \varphi_i \right)}, \,\,\,\,\,\,\,\,\,\,\,\,\,\,\,\,\,\,\,\,\,\,\,\,\,\,\,\,\,\,\,\,\,\,\,\,\,\,\,\,\,\,\,\,\,\,\,\,\,\,\,\,\,\,\,\,\,\,\,\,\,\,\,\,
      \end{equation}
      from which we have 
      \begin{equation}\label{Eq_Energy_Power_transfrom}
         \left(\frac{1 - \theta_i \varphi_i}{1 - \varphi_i}\right) \tau_i P_i - \theta_i \sum\limits_{\scriptstyle j = 1 \hfill \atop \scriptstyle j \ne i \hfill}^K {\tau_j P_j H_{j,i}} = \theta_i P_0 H_{0,i},
      \end{equation}
      where $\theta_i = \eta_i \zeta_i$. It then follows that 
      \begin{equation}\label{Eq_TransmissionEnergy_Matrix}
         {\bf{A}}{\boldsymbol{\Omega}}{\bf{P}} = P_0 \bf{b},
      \end{equation}
      where $P = [P_1, P_2, \cdots P_K]^T$ and $\boldsymbol{\Omega}$ $=$ $\rm{diag} \{[\tau_1, \tau_2, \cdots, ¥óK]^T\}$ with $\rm{diag} \{\bf{v}\}$ denoting a diagonal matrix, of which the diagonal entries consist of vector $\bf{v}$. Furthermore, $\bf{A}$ is a matrix with
      \begin{equation}\label{Eq_Matrix_A}
         {A_{i,j}} = \left\{ {\begin{array}{*{20}{c}}
         {\left(1 - {\theta _i}{\varphi_i}\right) / \left( 1 - \varphi_i \right)}  \\
         { - {\theta _i}{H_{i,j}}}  \\
         \end{array}\begin{array}{*{20}{c}}
         {,\,\,\, \,\,\, {\rm{if}}\,\,\,j = i \,\,\, }  \\
         {,\,\,\,{\rm{otherwise,}}}  \\
         \end{array}} \right.
      \end{equation}
      where $A_{i,j}$ denotes the element of matrix {\bf{A}} on the $i$-th row and $j$-th column. Finally, $\bf{b}$ is a vector given by
      \begin{equation}\label{Eq_Vector_Beta}
         {\bf{b}} = \left[ \theta_1 H_{0,1} \,\, \theta_2 H_{0,2} \,\, \cdots \,\, \theta_K H_{0,K} \right]^T.
      \end{equation}
      From (\ref{Eq_TransmissionEnergy_Matrix}) through (\ref{Eq_Vector_Beta}), $P_i$ is then given by 
      \begin{equation}\label{Eq_Pi}
         P_i = \rho_i \frac{P_0}{\tau_i}, \,\,\, i = 1, \,\, \cdots, \,\, K,
      \end{equation}
      where $\rho_i$, $i = 1, \cdots , K$, is given by 
      \begin{equation}\label{Eq_Rho}
         \left[\rho_1 \,\, \rho_2 \,\, \cdots \,\, \rho_K\right] = {\bf{A}}^{-1} {\bf{b}}.
      \end{equation}
      
      Note that at the H-AP, the residual SI after SIC can be approximated as $I_0 \sim \mathcal{CN} (0, \alpha P_0)$, where $¥á \ll 1$ \cite{Day}. Given the time allocation $\boldsymbol{\tau} = [\tau_1, \cdots, \tau_K]$, the achievable UL throughput of $U_i$ during the $i$-th slot is then given from (\ref{Eq_RxSignal_AP}) and (\ref{Eq_Pi}) by
      \[
         R_i \left( {\boldsymbol{\tau}} \right) = \tau_i \log_2 \left( 1 + \frac{H_{0,i} P_i}{\Gamma \left(\sigma_0^2 + \alpha P_0\right)} \right) \,\,\,\,\,\,\,\,\,\,\,\,\,\,\,\,\,\,\,\,\,\,\,\,\,\,\,\,\,\,\,\,\,\,\,\,\,\,\,\,\,\,
      \]
      \begin{equation}\label{Eq_UL_Rate}
         \,\,\,\, = \tau_i \log_2 \left( 1 +  \frac{\gamma_i \left( P_0 \right)}{\tau_i} \right), \,\,\, i = 1, \,\, \cdots, \,\, K,
      \end{equation}
      where $\gamma_i$ is given by
      \begin{equation}\label{Eq_gamma}
         \gamma_i \left( P_0 \right) = \frac{\rho_i H_{i,0} P_0 }{\Gamma \left( \sigma_0^2 + \alpha P_0 \right)}, \,\,\, i = 1, \,\, \cdots, \,\, K,
      \end{equation}
      with $\rho_i$, $i = 1, \cdots , K$, given in (\ref{Eq_Rho}), and $\Gamma$ representing the gap of signal-to-interference-plus noise ratio (SINR) from the additive white Gaussian noise (AWGN) channel owing to practical modulation and coding used.
      
      \subsubsection{Practical throughput achievable by partial knowledge of channels}
      In practice, it is difficult for the H-AP to obtain the knowledge of the channels between different UEs, in other words, $H_{i,j}$, $i, j \in \{0, \cdots, K\}$, $i \ne j$, $i \ne 0$, $j \ne 0$. At a UE, furthermore, energy harvested from the transmissions of other UE is small enough to be negligible \cite{Ju: HD-WPCN}, \cite{Ju: FD-WPCN}. Therefore, a practical throughput of a FD-WPCN-FD can be obtained by assuming that the H-AP only has the knowledge of $h_{0,i}$, and $\varphi_i$, $\forall i \in \{0, \cdots, K\}$, and by ignoring the energy harvested from the transmissions of other UEs, in other words,
      \begin{equation}\label{Eq_HarvestedEnergy_otherUE}
         \sum\limits_{j = 1,\,\,j \ne i}^K  {\tau _j}{P_j}{H_{j,i}}T\,\, = 0,
      \end{equation}
      in (\ref{Eq_Energy_Ui}). In this case, the amount of harvested energy at $U_i$ in (\ref{Eq_Energy_Ui_new2}) is modified as
      \begin{equation}\label{Eq_HarvestedEnergy_Prac}
         E_i^h = {\zeta _i}\left( {{P_0}{H_{0,i}} + {\tau _i}{P_i}{H_{i,i}}} \right)T,
      \end{equation}
      By the same procedure from (\ref{Eq_LeakageEnergy}) through (\ref{Eq_TransmissionEnergy_Matrix}) with $E_i^h$ replaced by (\ref{Eq_HarvestedEnergy_Prac}), the matrix $\bf{A}$ described in (\ref{Eq_Matrix_A}) is then modified as a diagonal matrix given by
      \begin{equation}\label{Eq_Matrix_A_Prac}
         {A_{i,j}} = \left\{ {\begin{array}{*{20}{c}}
         {\left( {1 - {\theta _i}{\varphi _i}} \right)/\left( {1 - {\varphi _i}} \right)}\\
         0
         \end{array}\begin{array}{*{20}{c}}
         {,{\kern 1pt} {\kern 1pt} {\kern 1pt} {\kern 1pt} \,\,{\kern 1pt} {\rm{if}}{\kern 1pt} {\kern 1pt} {\kern 1pt} j = i\,\,\,,}\\
         {,{\kern 1pt} {\kern 1pt} {\kern 1pt} {\rm{otherwise}}{\rm{.}}}
         \end{array}} \right.
      \end{equation}
      Accordingly, $\rho_i$ , $i = 1, \cdots, K$, in (\ref{Eq_Pi}) is obtained by (\ref{Eq_Rho}) and (\ref{Eq_Matrix_A_Prac}) as
      \begin{equation}\label{Eq_Rho_Prac}
         {\rho _i} = \frac{{\left( {1 - {\varphi _i}} \right)}}{{\left( {1 - {\theta _i}{\varphi _i}} \right)}}{\theta _i}{H_{0,i}}.
      \end{equation}
      The throughput of FD-WPCN-FD in this practical sense is then given by (\ref{Eq_UL_Rate}) and (\ref{Eq_gamma}) where $\rho_i$ in (\ref{Eq_gamma}) is replaced by that in (\ref{Eq_Rho_Prac}). As clearly shown in (\ref{Eq_UL_Rate}), (\ref{Eq_gamma}), and (\ref{Eq_Rho_Prac}), throughput of the FD-WPCN-FD increases with $\theta_i$, in other words, increase of $\eta_i$ or $\zeta_i$.

   \subsection{Weighted Sum-Throughput Maximization}\label{TimeAllocation}
   As shown in (\ref{Eq_UL_Rate}), the throughput of the FD-WPCN-FD can be maximized by optimizing the time allocated to each UE for a UL information transmission, denoted by $\boldsymbol{\tau} = [\tau_1, \cdots, \tau_K]$. Specifically, we maximize the weighted sum-throughput in this network by solving the problem formulated as follows:
   \begin{equation}\label{Eq_Prop_Primal}
      {\rm{(P1)}}: \,\,\,\, \mathop {\max }\limits_{\boldsymbol{\tau}}  \,\,\,\sum\limits_{i = 1}^K { \omega_i {R_i}\left( \boldsymbol{\tau}  \right)} \,\,\,\,\,\,\,\,\,\,\,\,\,\,\,\,\,\,\,\,
   \end{equation}
   \begin{equation}\label{Eq_Prop_Constraint1}
      {\rm{s.t}} \,\,\,\,\, \sum\limits_{i=1}^K {\tau_i} \le 1,
   \end{equation}
   \begin{equation}\label{Eq_Prop_Constraint2}
      \,\,\,\,\,\,\,\,\,\,\,\,\,\,\,\,\,\,\,\,\,\,\,\,\,\,\,\,\,\,\,\,\,\,\,\,\,\,\,\,\,\,\,\,\,\,\, \tau_i > 0, \,\,\, i = 1, \,\, \cdots, \,\, K.
   \end{equation}
   
   It can be easily seen that problem (P1) is a convex optimization problem because $R_i (\boldsymbol{\tau})$ and (\ref{Eq_Prop_Constraint1}) are concave and affine functions of $\boldsymbol{\tau}$, respectively. The optimal time allocation solution for (P1), denoted by $\boldsymbol{\tau}^*$, is then obtained as follows.   

   \begin{proposition}\label{Proposition_Opt_Time_Alloc}
      The weighted sum-throughput of the FD-WPCN-FD is maximized by the optimal time allocation solution for (P1), given by ${\boldsymbol{\tau}}^* = [\tau_1^* \,\, \cdots \,\, \tau_K^*]$ with
      \begin{equation}\label{Eq_Prop_Opt_tau}
         \tau_i^* = \frac{\gamma_i \left( P_0 \right)}{z_i^*}, \,\,\, i = 1, \,\, \cdots, \,\, K,
      \end{equation}
      where $z_i^*$ denote the solution of $f\left( z \right) = \frac{\lambda^* \ln 2}{\omega_i}$, with $f\left( z \right)$ defined as
      \begin{equation}\label{Eq_function_f}
         f\left( z \right) \buildrel \Delta \over = \ln \left( 1 + z \right) - \frac{z}{1 + z}.
      \end{equation}
   \end{proposition}
   \begin{proof}
      The proof is shown in Appendix \ref{Proof_Proposition}.
   \end{proof}
   
   Figure \ref{Fig_RateRegion} compares the achievable throughput region of the FD-WPCN-FD against those of the HD-WPCN and FD-WPCN-HD studied in \cite{Ju: FD-WPCN}, with $K = 2$, $P_0 = 20$ dB, $\sigma_0^2 = 0$ dB, and $\Gamma = 1$. Duplex modes of the H-AP and UEs in the HD-WPCN, FD-WPCN-HD, and FD-WPCN-FD are summarized in Table \ref{Table_DuplexModes}. It is assumed that $H_{0,1} = 0.50$, $H_{0,2} = 0.15$, and $H_{1,2} = H_{2,1} = 0.01$. For the FD-WPCN-HD, SIC at the H-AP is assumed to be perfect. For FD UEs in the FD-WPCN-FD, it is assumed that $\varphi_1 = \cdots = \varphi_K = 0.03$ with 15 dB of isolation between transmission and reception. Considering practical difficulty for the H-AP to obtain the knowledge of channels between different UEs, for all considered WPCNs we assume that the H-AP has the knowledge of $H_{0,1}$ and $H_{0,2}$ only, and the energy harvested from the transmissions of other UEs is thus ignored. In case of FD-WPCN-FD, however, the ideal throughput obtained by (\ref{Eq_Energy_Ui}) through (\ref{Eq_gamma}) with the knowledge of $H_{1,2}$ and $H_{2,1}$ is also presented for the purpose of comparison.

   \begin{figure}[!t]
      \centering
      \includegraphics[width=1.0\columnwidth]{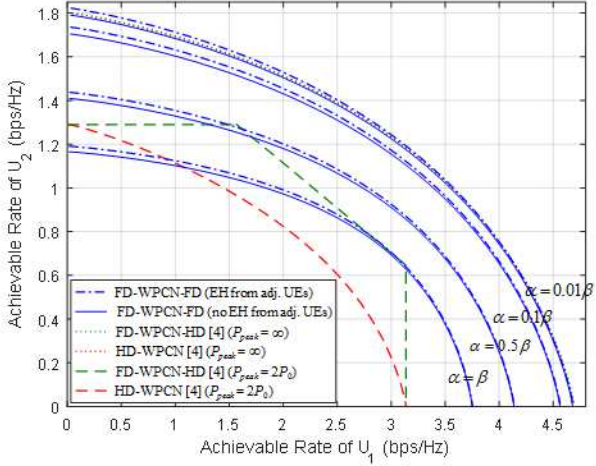}
      \caption{Rate regions of WPCNs according to duplex modes.}
      \label{Fig_RateRegion}
   \end{figure}

   \begin{table}\center\abovecaptionskip -0.6\baselineskip
      \caption{Comparison of duplex modes in various WPCNs}\center
      \begin{tabular}{|c|c|c|}
         \hline
         \hline
            & H-AP & UE \\                
         \hline
            HD-WPCN & HD (TDD) & HD (TDD) \\
         \hline
            FD-WPCN-HD & FD & HD (TDD) \\
         \hline
            FD-WPCN-FD & FD & FD \\           
         \hline
         \hline
      \end{tabular}
      \label{Table_DuplexModes}
   \end{table}

   It is shown from this figure that the throughput region of the FD-WPCN-FD is enlarged as $\alpha$ decreases (in other words, SIC performance improves). In addition, the throughput of FD-WPCN-FD is observed to be increased by taking the energy harvested from the other UEs' transmissions into account based on full knowledge of channels, as compared to the counterpart that ignores this energy. However, the increase is insignificant.
   
   This figure also compares the throughput region of FD-WPCN-FD obtained by (\ref{Eq_UL_Rate}) through (\ref{Eq_Rho_Prac}) with those of HD-WPCN, FD-WPNC-HD, when the energy harvested from other UEs' transmissions is not considered. It is observed from this figure that the throughput region of the FD-WPCN-FD when $\alpha = 0.5\beta$ with $\beta = \sigma_0^2 / P_0$ is shown to be larger than those of HD-WPCN and FD-WPCN-HD with $P_{\rm{peak}} = 2P_0$, where $P_{\rm{peak}}$ denotes the maximum peak transmit power at the H-AP in HD-WPCN and FD-WPCN-HD. When $\alpha = 0.01\beta$, in addition, the throughput region of the FD-WPCN-FD approaches those of HD-WPCN and FD-WPCN-HD with $P_{\rm{peak}} = \infty$, which are theoretically maximum achievable throughput regions of HD-WPCN and FD-WPCN-HD.

   \begin{corollary}\label{Corollary_Time_Alloc_SumRateMax}
      The sum-throughput of the FD-WPCN-FD is maximized with a time allocation ${\boldsymbol{\tau}}^{\star} = [\tau_1^{\star} \,\, \cdots \,\, \tau_K^{\star}]$, with
      \begin{equation}\label{Eq_Corr_Opt_tau}
         \tau_i^{\star} = \frac{\gamma_i \left( P_0 \right) }{\sum\limits_{j=1}^K {\gamma_j \left( P_0 \right)} }, \,\,\, i = 1, \,\, \cdots, \,\, K.
      \end{equation}
   \end{corollary}
   \begin{proof}
      The proof is shown in Appendix \ref{Proof_Corollary}.
   \end{proof}

\section{Simulation Results}\label{SimulationResult}
In this section, the sum-throughput of the HD-WPCN, FD-WPCN-HD, and FD-WPCN-FD are compared through simulations, with $K = 10$ and $\theta_i = 0.5$, $\forall i = 1, \cdots, K$. Duplex modes of the H-AP and UEs in the aforementioned WPCNs follow those described in Table \ref{Table_DuplexModes}. We compare the practical throughput of these WPCNs by ignoring the energy harvested from the transmissions of other UEs. For the FD-WPCN-FD, however, we also present the ideal throughput taking the energy harvested from the transmissions of other UEs into account for the purpose of comparison. The bandwidth and noise spectral density are set as 1 MHz and -160 dBm/Hz, respectively. The UEs are assumed to be uniformly distributed within two concentric circles, with diameters of 5 and 10 m, respectively, where the H-AP is located at the center of these concentric circles. Assuming that the average signal power attenuation at a reference distance of 1 m is 30 dB, the channel power gain $H_{i,j}$ is modeled as $H_{i,j} = \nu_{i,j}10^{-3}D_{i,j}^{-\delta}$, $\forall i, j = \{0, \cdots, K\}$, $i \ne j$, with $D_{i,j}$ denoting the distance between $U_i$ and $U_j$ in meters, and the pathloss exponent given by $\delta = 2$. We also assume that the short-term fading, denoted by $\nu$, is Rayleigh distributed, in other words, $\nu$ is an exponentially distributed random variable with unit mean. For FD UEs in FD-WPCN-FD, $\varphi_i = 0.03$, $\forall i = 1, \cdots, K$, assuming 15 dB of isolation between transmission and reception. Finally, we set $\Gamma = 9.8$ dB assuming an uncoded quadrature amplitude modulation with the required bit-error rate given by $10^{-7}$ \cite{Goldsmith}.

\begin{figure}[!t]
   \centering
   \includegraphics[width=1.0\columnwidth]{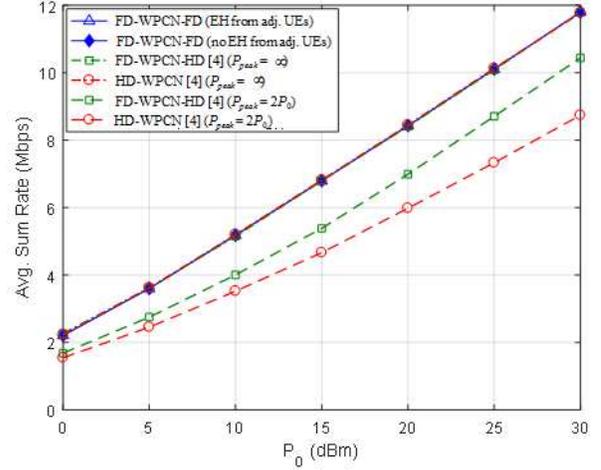}
   \caption{Sum-throughput vs. $P_0$.}
   \label{Fig_SumRate_vs_P0}
\end{figure}

In Fig. \ref{Fig_SumRate_vs_P0}, the maximum sum-throughput of the HD-WPCN, FD-WPCN-HD, and FD-WPCN-FD are compared with respect to the $P_0$ values in dBm. For both ideal and practical scenarios in FD-WPCN-FD, the throughput is obtained by assuming $\alpha = 0.01\beta$ with $\beta$ defined in the previous section, whereas perfect SIC is assumed for the FD-WPCN-HD.

We can see in this figure that in the FD-WPCN-FD, the difference between the sum throughput obtained with and without taking the energy harvested from transmissions of other UEs' into account is negligible. This verifies that the energy harvested from the transmissions of other UEs has little effects on the increase of throughput of FD-WPCN-FD. Furthermore, this implies that the throughput of the FD-WPCN-FD can be optimized even with simpler approach presented in (\ref{Eq_UL_Rate})-(\ref{Eq_Rho_Prac}).

In addition, it is also observed in this figure that even when the energy harvested from other UEs' transmissions is not taken into account, the average sum-throughput of the FD-WPCN-FD is equivalent to those of the HD-WPCN and FD-WPCN-HD with $P_{\rm{peak}} = \infty$. Note that the throughput of HD-WPCN and FD-WPCN-HD with $P_{\rm{peak}} = \infty$ is an ideal throughput that is achievable only theoretically, but not achievable in practice. Given SIC capability at the H-AP which reduces the power of residual SI after SIC sufficiently close to the level of background noise, the FD-WPCN-FD can achieve the throughput equivalent to the theoretically maximum ones of HD-WPCN and FD-WPCN-HD, even with practical settings and simple resource allocation algorithm.

As compared to the FD-WPCN-HD with $P_{\rm{peak}} = 2P_0$, finally, the FD-WPCN-FD achieves a sum-throughput gain of 4 dB with respect to $P_0$. Furthermore, the sum-throughput of the FD-WPCN-FD is shown to increase more quickly with the increase in $P_0$ than that of the HD-WPCN with $P_{\rm{peak}} = 2P_0$.

\begin{figure}[!t]
   \centering
   \includegraphics[width=1.0\columnwidth]{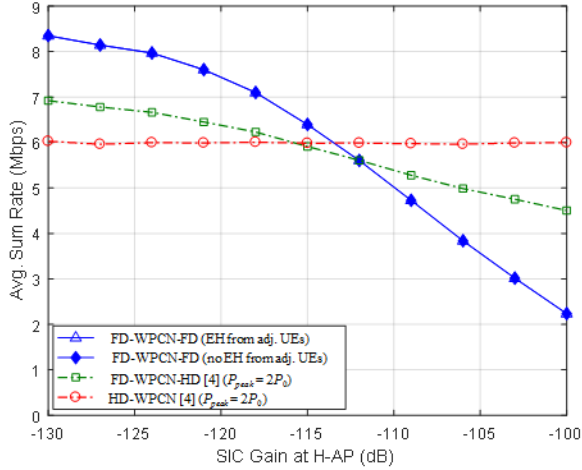}
   \caption{Sum-throughput vs. SIC gain at the H-AP.}
   \label{Fig_SumRate_vs_SIC}
\end{figure}

Fig. \ref{Fig_SumRate_vs_SIC} compares the maximum sum-throughput of the aforementioned WPCNs for different values of SIC gain, where the SIC gain is defined as $1/\alpha$. The sum-throughput of the FD-WPCN-FD is shown to be larger than those of the HD-WPCN and FD-WPCN-HD when SIC gain at the H-AP is larger than 114 dB, in other words, the power of the residual SI after SIC is $4\sigma_0^2$. This implies that the use of the FD UEs shown in Fig. \ref{Fig_UE_Transceiver} efficiently improves the throughput of the WPCN even with current SIC technologies at the H-AP (for example, \cite{KUMU} and \cite{Chea}), which reduce the power of the residual SI sufficiently close to $\sigma_0^2$ after SIC. When $\alpha$ is 120dB, in other words, the aggregated power of residual SI and background noise is $4\sigma_0^2$, FD-WPCN-FD achieves $18\%$ and $25\%$ of throughput gain as compared to FD-WPCN-HD and HD-WPCN, respectively.

\begin{figure}[!t]
   \centering
   \includegraphics[width=1.0\columnwidth]{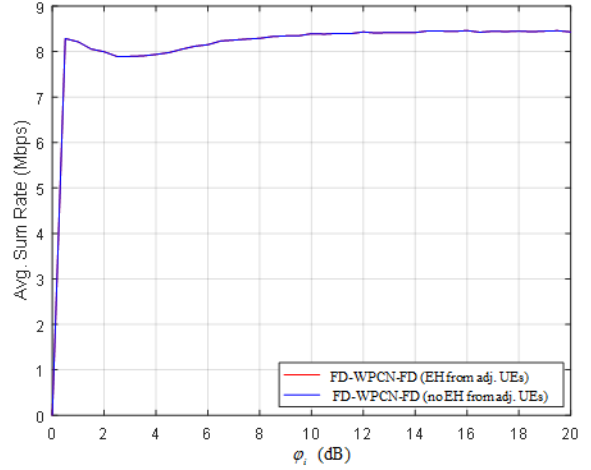}
   \caption{Sum-throughput vs. $\varphi_i$ at the UE.}
   \label{Fig_SumRate_vs_EH}
\end{figure}

Finally, Fig. \ref{Fig_SumRate_vs_EH} shows the throughput of the FD-WPCN-FD with respect to the circulator leakage $\varphi_i$. When $\varphi_i = 0$ dB, throughput of the FD-WPCN-FD is zero because whole energy at PA output leaks into the energy harvester and thus information is not transmitted. When $\varphi_i \ne 0$, a trade-off associated with the value of $\varphi_i$ is observed. Throughput of FD-WPCN FD decreases with increasing $\varphi_i$ when $\varphi_i \le 3$ dB, while it increases with $\varphi_i$ when $\varphi_i \ge 3$ 3 dB. This shows that for increase of throughput in the FD-WPCN-FD, the increase of the harvested energy thanks to larger SI is dominant when $\varphi_i$ is small, while utilizing more energy for WIT with smaller $\varphi_i$ is more beneficial when $\varphi_i$ is large enough although the amount of the harvested energy decreases. However, the effect of the trade-off associated with $\varphi_i$ on the throughput of FD-WPCN-FD is observed to be insignificant as long as $\varphi_i \ne 0$. Finally, the throughput of the FD-WPCN-FD converges when $\varphi_i \ge 14$ dB.

\section{Conclusion} \label{Conclusion}
This paper described a WPCN where both the H-AP and UEs operate in FD mode. Thanks to the FD capability, UEs can receive energy and transmit information simultaneously in the DL and UL, respectively. We first proposed an FD transceiver structure for UEs that enables simultaneous energy reception and information transmission, and then provided a energy usage model in the proposed UE transceiver. Furthermore, an optimal time allocation solution to maximize the weighted sum-throughput of the WPCN with the FD H-AP and FD UEs was obtained for both cases where the energy harvested from other UEs' transmissions is and is not taken into account. It was shown that the use of the proposed FD UEs as well as an FD H-AP with practical SIC capability significantly improves the throughput of the WPCN. In addition, the use of FD H-AP and FD UEs is shwon to make WPCN achieve the throughput which is equivalent to the theoretically maximum ones of conventional HD-WPCN and FD-WPCN-HD, even with practical settings and simple resource allocation algorithm.

\appendices
\section{Proof of Proposition \ref{Proposition_Opt_Time_Alloc}}\label{Proof_Proposition}
From (\ref{Eq_Prop_Primal})-(\ref{Eq_Prop_Constraint2}), we have the Lagrangian of $(\rm{P1})$ as follows.
\begin{equation}\label{App_P1_Lagrangian}
   {\mathcal{L}\left( {\boldsymbol{\tau} ,\lambda } \right) = \sum\limits_{i = 1}^K {\omega_i {R_i}\left( \boldsymbol{\tau}  \right)}  - \lambda \left( {\sum\limits_{i = 1}^K {{\tau _i}} } - 1 \right),}
\end{equation}
with $\lambda \ge 0$ representing the Lagrange multiplier associated with (\ref{Eq_Prop_Constraint1}). We thus have the dual function of ($\rm{P1}$) as follows.
\begin{equation}\label{App_DualFunction}
   {\mathcal{G}\left( \nu  \right) = \mathop {\min }\limits_{\boldsymbol{\tau}  \in \mathcal{D}} \,\,{\mathcal{L}}\left( {\boldsymbol{\tau} ,\lambda } \right),}
\end{equation}
with $\mathcal{D}$ denoting the feasible set of $\boldsymbol{\tau}$ specified by (\ref{Eq_Prop_Constraint1}) and (\ref{Eq_Prop_Constraint2}). Note that Slater's condition holds for (P1) because we can find a $\boldsymbol{\tau} \in \mathcal{D}$, $\tau_i > 0$, $i = 1, \,\, \cdots \,\, K$, with which $\sum\nolimits_{i = 1}^K {{\tau _i} < 1}$. Therefore, $(\rm{P1})$ is shown to be a convex optimization problem for which the strong duality holds, and the globally optimal solution for (P1) thus satisfies the Karush-Kuhn-Tucker (KKT) conditions \cite{Boyd}, given by
\begin{equation}\label{App_OptimalityCondition1}
   {\sum\limits_{i = 1}^{K} {\tau _i^*}  \le 1,}
\end{equation}
\begin{equation}\label{App_OptimalityCondition2}
   {\lambda^* \left( {\sum\limits_{i = 1}^K {{\tau _i^*}} } - 1 \right) = 0,}
\end{equation}
\begin{equation}\label{App_OptimalityCondition3}
   {\frac{\partial }{{\partial {\tau _i}}} \sum\limits_{i = 1}^K {{R_i}\left( \boldsymbol{\tau}^*  \right)} - \lambda^* = 0, \,\,\,\, i = 0, \,\, \cdots, \,\,K,}
\end{equation}
with $\tau_i^*$'s and $\lambda^*$ denoting, respectively, the optimal primal and dual solutions of ($\rm{P1}$). With no loss of generality, we can assume $\lambda^* > 0$ from (\ref{App_OptimalityCondition2}) because $\sum\nolimits_{i = 1}^K {{\tau _i^*} = 1}$ must hold for ($\rm{P1}$). From (\ref{App_OptimalityCondition3}), it then follows that
\begin{equation}\label{App_OptimalityCondition4}
   { \ln \left( 1 + z_i^* \right) } - \frac{z_i^*}{1 + z_i^*} = \frac{\lambda^* \ln 2}{\omega_i}, \,\, i = 1 \,\, \cdots, \,\, K,
\end{equation}
where $z_i^* = \gamma_i(P_0) / \tau_i^*$. Note that $f \left( z \right)$ in (\ref{Eq_function_f}) is a monotonically increasing function of $z$ with $f\left( 0 \right) = 0$ and $f\left( \infty \right) = \infty$. Therefore, ${\boldsymbol{\tau}}^*$ is uniquely obtained for given values of $\gamma_i \left( P_0 \right)$. The proof of Proposition \ref{Proposition_Opt_Time_Alloc} is completed.

\section{Proof of Corollary \ref{Corollary_Time_Alloc_SumRateMax}}\label{Proof_Corollary}
The sum-throughput of this network is maximized by solving (P1) with $\omega_1 = \omega_2 = \cdots, \,\, \omega_K = \frac{1}{k}$. For the inequality in (\ref{App_OptimalityCondition4}) to hold in this case, we should have
\begin{equation}\label{App_TimePortionRatio}
   {\frac{{{\gamma_1}}\left( P_0 \right)}{{\tau _1^{\star}}} = \frac{{{\gamma_2}} \left( P_0 \right)}{{\tau _2^{\star}}} =  \cdots \frac{{{\gamma_K}} \left( P_0 \right)}{{\tau _K^{\star}}}.}
\end{equation}
From (\ref{App_OptimalityCondition2}) and (\ref{App_TimePortionRatio}), it follows that
\begin{equation}\label{Eq_App_SumTime}
      \sum\limits_{i = 1} ^K {\tau_i^{\star}} = \frac{\tau_i^{\star}}{\gamma_i} \sum\limits_{j = 1}^K {\gamma_j} = 1,
\end{equation}
From (\ref{Eq_App_SumTime}), the optimal time allocation solution in (\ref{Eq_Corr_Opt_tau}) is obtained. Corollary \ref{Proposition_Opt_Time_Alloc} is now proved.

\end{document}